\documentclass[aps,pra,showpacs]{article}

\usepackage{authblk}
\usepackage[T1]{fontenc}
\usepackage[latin1]{inputenc}
\usepackage[normalem]{ulem}
\usepackage{subeqnarray}
\usepackage{enumitem}
\usepackage{fullpage}
\usepackage[cmex10,intlimits]{amsmath}
\usepackage{amssymb}
\usepackage{array}
\usepackage{mdwmath}
\usepackage{mdwtab}
\usepackage{amsthm}
\usepackage{amsmath}
\usepackage{cite}
\usepackage[pdftex]{graphicx,hyperref,color}
\usepackage{upgreek}
\graphicspath{{figures/}}

\newtheorem{theorem}{Theorem}[section]

\title{Generation and Distribution of Quantum Oblivious Keys for Secure Multiparty Computation}
\author[1,2]{Mariano Lemus}
\author[3,4]{Mariana F. Ramos}
\author[1,2]{Preeti Yadav}
\author[3,4]{Nuno A. Silva}
\author[3,4]{Nelson J. Muga}
\author[5]{Andr\'e Souto}
\author[1,2]{Nikola Paunkovi\'c}
\author[1,2]{Paulo Mateus}
\author[3,4]{Armando N. Pinto}

\affil[1]{Instituto de Telecomunica{c}\~{o}es, Lisbon, Portugal}
\affil[2]{Departamento de Matem\'{a}tica, Instituto Superior T\'{e}cnico, Av. Rovisco Pais, Lisbon, Portugal}
\affil[3]{Instituto de Telecomunica{c}\~{o}es, University of Aveiro, Campus Universit\'{a}rio de Santiago, 3810-193, Aveiro, Portugal}
\affil[4]{Department of Electronics, Telecommunications and Informatics, University of Aveiro, Portugal}
\affil[5]{Departamento de Inform\'{a}tica, Faculdade de Ci\^{e}ncias da Universidade de Lisboa, Lisbon, Portugal}

\begin{document}
\maketitle
\begin{abstract}%
The oblivious transfer primitive is sufficient to implement secure multiparty computation. However, secure multiparty computation based on public-key cryptography is limited by the security and efficiency of the oblivious transfer implementation. We present a method to generate and distribute oblivious keys by exchanging qubits and by performing commitments using classical hash functions. With the presented hybrid approach, quantum and classical, we obtain a practical and high-speed oblivious transfer protocol. We analyse the security and efficiency features of the technique and conclude that it presents advantages in both areas when compared to public-key based techniques.
\end{abstract}

\section{Introduction}

In Secure Multiparty Computation (SMC), several agents compute a function that depends on their own inputs, while maintaining them private~\cite{Lindell09}. Privacy is critical in the context of an information society, where data is collected from multiple devices (smartphones, home appliances, computers, street cameras, sensors, ...) and subjected to intensive analysis through data mining. This data collection and exploration paradigm offers great opportunities, but it also raises serious concerns. A technology able to protect the privacy of citizens, while simultaneously allowing to profit from extensive data mining, is going to be of utmost importance. SMC has the potential to be that technology if it can be made practical, secure and ubiquitous.

Current SMC protocols rely on the use of asymmetric cryptography algorithms~\cite{Laud2015}, which are considered significantly more computationally complex compared with symmetric cryptography algorithms~\cite{Asharov2017}. Besides being more computationally intensive, in its current standards, asymmetric cryptography cannot be considered secure anymore due to the expected increase of computational power that a large-scale quantum computer will bring~\cite{Bernstein2017}. Identifying these shortcomings in efficiency and security motivates the search for alternative techniques for implementing SMC without the need of public key cryptography.

\subsection{Secure Multiparty Computation and Oblivious Transfer}

Consider a set of $N$ agents and $f(x_1, x_2, ..., x_N)=(y_1, y_2, ..., y_N)$ a multivariate function. For $i \in \lbrace 1,...,N \rbrace$, a SMC service (see Figure~\ref{fig:smc_N}) receives the input $x_i$ from the $i$-th agent and outputs back the value $y_i$ in such a way that no additional information is revealed about the remaining $x_j, y_j$, for $j \neq i$. Additionally, this definition can be strengthened by requiring that for some number $M < N$ of corrupt agents working together, no information about the remaining agents gets revealed (secrecy). It can also be imposed that if at most $M'<N$ agents do not compute the function correctly, the protocol identifies it and aborts (authenticity).

\begin{figure}
  \centering
  \includegraphics[width=0.45\linewidth]{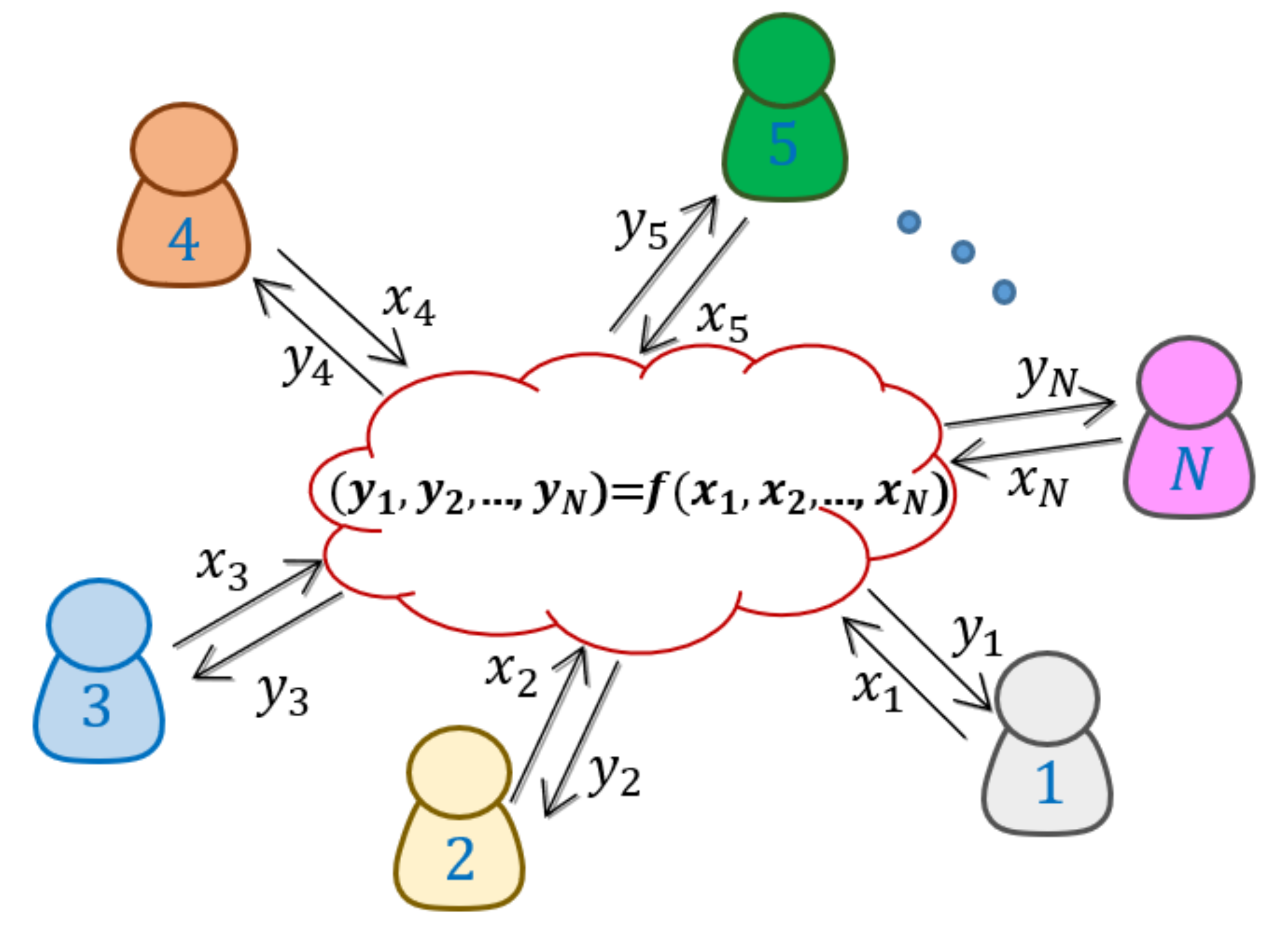}
  \caption{In secure multiparty computation, N parties compute a function preserving the privacy of their own input. Each party only has access to their own input-output pair.}
  \label{fig:smc_N}
\end{figure}

Some of the most promising approaches towards implementing SMC are based on oblivious circuit evaluation techniques such as Yao's garbled circuits for the two party case~\cite{4568207} and the GMW or BMR protocols for the general case~\cite{Goldreich:1987:PAM:28395.28420, Schneider2013, Laud2015, Beaver1990}. It has been shown that to achieve SMC it is enough to implement the Oblivious Transfer (OT) primitive and, without additional assumptions, the security of the resulting SMC depends only on that of the OT~\cite{Kilian:1988:FCO:62212.62215}. In the worst case, this requires each party to perform one OT with every other party for each gate of the circuit being evaluated. This number can be reduced by weakening the security or by increasing the amount of exchanged data~\cite{Harnik07}. Either way, the OT cost of SMC represents a major bottleneck for its practical implementation. Finding fast and secure OT protocols, hence, is a very relevant task in the context of implementing SMC. 

Let Alice and Bob be two agents. A 1-out-of-2 OT service receives bits $b_0, \, b_1$ as input from Alice and a bit $c$ as input from Bob, then outputs $b_c$ to Bob. This is done in a way that Bob gets no information about the other message, i.e. $b_{\overline{c}}$, and Alice gets no information about Bob's choice, i.e. the value of $c$~\cite{eprint-2005-12523}.

\subsection{State of the art}

Classical OT implementations are based on the use of asymmetric keys, and suffer from two types of problems. The first one is the efficiency: asymmetric cryptography relies on relatively complex key generation, encryption, and decryption algorithms ~\cite[Chapter~1]{Goldreich2001}~\cite [Chapter 6]{Paar10}. This limits achievable rates of OTs, and since implementations of SMC require a very large number of OTs ~\cite{Harnik07}~\cite{Asharov2017}, this has hindered the development of SMC-based applications. The other serious drawback is that asymmetric cryptography, based on integer number factorization or discrete-logarithm problems, is insecure in the presence of quantum computers, and therefore, it has to be  progressively abandoned. There are strong research efforts in order to find other hard problems that can support asymmetric cryptography~\cite{Bernstein2017}. However, the security of these novel solutions is still not fully understood. 

A possible way to circumvent this problem is by using quantum cryptography to improve the efficiency and security of current techniques. Quantum solutions for secure key distribution, Bit Commitment (BC) and OT have been already proposed~\cite{Broadbent2016}. The former was proved to be unconditionally secure (assuming an authenticated channel) and realizable using current technology. Although, it was shown to be impossible to achieve unconditionally secure quantum BC and OT ~\cite{Shenoy17}~\cite{lo:chau:97}~\cite{Mayers97}, one can impose restrictions on the power of adversaries in order to obtain practically secure versions of these protocols~\cite{PhysRevLett.100.220502,PhysRevA.81.052336}. These assumptions include physical limitations on the apparatuses, such as noisy or bounded quantum memories~\cite{6157089,lou:14,Almeida_2015}. For instance, quantum OT and BC protocols have been developed and implemented (see~\cite{2014NatCo3418E,2018NatCo450F,2012NatCo1326N}) under the noisy storage model. Nevertheless, solutions based on hardware limitations may not last for long, because as quantum technology improves the rate of secure OT instances will decrease. Other solutions include exploring relativistic scenarios using the fact that no information can travel faster than light \cite{PhysRevLett.115.030502,PhysRevLett.117.140506,PhysRevA.98.032327}. However, at the moment, these solutions do not seem to be practical enough to allow the large dissemination of SMC.


In this work, we explore the resulting security and efficiency features of implementing oblivious transfer using a well known quantum protocol~\cite{yao:86} supported by using a cryptographic hash based commitment scheme~\cite{Halevi1996}. We call it a hybrid approach, since it mixes both classical and quantum cryptography. We analyse the protocol stand alone security, as well as its composable security in the random oracle model. Additionally, we study its computational complexity and compare it with the complexity of alternative public key based protocols. Furthermore, we show that, while unconditional information-theoretic security cannot be achieved, there is an advantage (both in terms of security and efficiency) of using quantum resources in computationally secure protocols, and as such, they are worth consideration for practical tasks in the near future. 

This paper is organized as follows. In Section II,  we present a quantum protocol to produce OT given access to a collision resistant hash function, define the concept of oblivious keys, and explain how having pre-shared oblivious keys can significantly decrease the computational cost of OT during SMC. The security and efficiency of the protocol is discussed in Section III. Finally, in Section IV we summarize the main conclusions of this work.


\section{Methods}

\subsection{Generating the OTs}

In this section, we describe how to perform oblivious transfer by exchanging qubits. The protocol $\pi_{QOT}$ shown in Figure~\ref{fig:okdprotocol} is the well known quantum oblivious transfer protocol first introduced by Yao, which assumes access to secure commitments. The two logical qubit states $|0\rangle$ and $|1\rangle$ represent the computational basis, and the states $|+\rangle = (|0\rangle + |1\rangle)/\sqrt{2}$, $|-\rangle = (|0\rangle - |1\rangle)/\sqrt{2}$ represent the Hadamard basis. We also define the states $|(s_i, a_i)\rangle$ for $s_i, a_i \in \lbrace 0,1 \rbrace$ according to the following rule: 
    \begin{align*}
          &|(0, 0)\rangle = |0\rangle \quad |(0, 1)\rangle = |+\rangle\\
          &|(1, 0)\rangle = |1\rangle \quad |(1, 1)\rangle = |-\rangle.
    \end{align*}
Note that these states can be physically instantiated using, for instance, a polarization encoding fiber optic quantum communication system, provided that a fast polarization encoding/decoding process and an algorithm to control random polarization drifts in optical fibers are available~\cite{Pinto2018, Ramos2020}. 

\begin{figure}[h]
\noindent
\framebox{\parbox{\dimexpr\linewidth-2\fboxsep-2\fboxrule}{\small
  {\centering \textbf{Protocol} $ \pi_{QOT}$ \par}
  \textbf{Parameters:} Integers $n, m<n$. \\
  \textbf{Parties:} The sender Alice and the receiver Bob. \\
  \textbf{Inputs:} Alice gets two bits $b_0,b_1$ and Bob gets a bit $c$. \\
  \textit{(Oblivious key distribution phase)}
  \begin{enumerate}[leftmargin=0.6cm]
      \item Alice samples $s, a \in \lbrace 0,1 \rbrace^{n+m}$. For each $i \leq n+m$ she prepares the state $|\phi_i\rangle = |(s_i, a_i)\rangle$ and sends $|\phi\rangle= |\phi_1\phi_2 \ldots \phi_{n+m}\rangle$ to Bob. \vspace{2mm}
      \item Bob samples $\Tilde{a} \in \lbrace 0,1 \rbrace^{n+m}$ and, for each $i$, measures $|\phi_i\rangle$ in the computational basis if $\Tilde{a}_i=0$, otherwise measures it in the Hadamard basis. Then, he computes the string $\Tilde{s}=\Tilde{s}_1\Tilde{s}_2 \ldots \Tilde{s}_{n+m}$, where $\Tilde{s}_i=0$ if the outcome of measuring $|\phi_i\rangle$ was $0$ or $+$, and $\Tilde{s}_i=1$ if it was $1$ or $-$. \vspace{2mm}
      \item For each $i$, Bob commits $(\Tilde{s}_i,\Tilde{a}_i)$ to Alice. \vspace{2mm}
      \item Alice chooses randomly a set of indices $T \subset \lbrace 1,\ldots, n+m \rbrace$ of size $m$ and sends $T$ to Bob. \vspace{2mm}
      \item For each $j \in T$, Bob opens the commitments associated to $(\Tilde{s}_j,\Tilde{a}_j)$. \vspace{2mm}
      \item Alice checks if $s_j = \Tilde{s}_j$ whenever $a_j = \Tilde{a}_j$ for all $j \in T$. If the test fails Alice aborts the protocol, otherwise she sends $a^* = a|_{\overline{T}}$ to Bob and sets $k=s|_{\overline{T}}$. \vspace{2mm}
      \item Bob computes $x= a^* \oplus \Tilde{a}|_{\overline{T}}$ and $\Tilde{k} = \Tilde{s}|_{\overline{T}}$. 
  \end{enumerate}
  \textit{(Oblivious transfer phase)}
  \begin{enumerate}[leftmargin=0.6cm]
      \setcounter{enumi}{7}
      \item Bob defines the two sets $I_0= \lbrace i \mid x_i=0 \rbrace$ and $I_1= \lbrace i \mid x_i=1 \rbrace$. Then, he sends to Alice the ordered pair $(I_{c},I_{c \oplus 1})$. \vspace{2mm}
      \item Alice computes $(e_0, e_1)$, where $e_i = b_i \bigoplus_{j \in I_{c\oplus i}} k_j$, and sends it to Bob. \vspace{2mm}
      \item Bob outputs $\Tilde{b}_c = e_c \bigoplus_{j \in I_{0}} \Tilde{k}_j$.
  \end{enumerate}
  }}
  \caption{Quantum OT protocol based on secure commitments. The $\bigoplus$ denotes the bit XOR of all the elements in the family.}
  \label{fig:okdprotocol}
\end{figure}

Intuitively, this protocol works because the computational and the Hadamard are conjugate bases. Performing a measurement in the preparation basis of a state, given by $a_i$, yields a deterministic outcome, whereas measuring in the conjugate basis, given by $\bar{a}_i$, results in a completely random outcome. By preparing and measuring in random bases, as shown in steps 1 and 2, approximately half of the measurement outcomes will be equal to the prepared states, and half of them will have no correlation. As Alice sends the information of preparation bases to Bob in step 6, he gets to know which of his bits are correlated with Alice's. During steps 3 to 6, Bob commits the information of his measurement basis and outcomes to Alice, who then chooses a random subset of them to test for correlations. Passing this test (statistically) ensures that Bob measured his qubits as stated in the protocol as opposed to performing a different (potentially joint) measurement. Such strategy may extract additional information from Alice's strings, but would fail to pass the specific correlation check in step 6. At step 8, Bob separates his non-tested measurement outcomes in two groups: $I_0$ where he measured in the same basis as the preparation one, and $I_1$, in which he measured in the different basis. He then inputs his bit choice $c$ by selecting the order in which he sends the two sets to Alice. During step 9, Alice encrypts her first and second input bits with the preparation bits associated with the first and second second sets sent by Bob respectively. This effectively hides Bob's input bit because she is ignorant about the measurements that were not opened by Bob (by the security of the commitment scheme). Finally, Bob can decrypt only the bit encrypted with the preparation bits associated to $I_0$.

In real implementations of the protocol one should consider imperfect sources, noisy channels, and measurement errors. Thus, in step 6 Alice should perform parameter estimation for the statistics of the measurements, and pass whenever the error parameter es below some previously fixed value. Following this, Alice and Bob perform standard post-processing techniques of information reconciliation and privacy amplification before continuing to step 7. These techniques indeed work even in the presence of a dishonest Bob. As long as he has some minimal amount of uncertainty about Alice's preparation string $s$, an adequate privacy amplification scheme can be used to maximize Bob's uncertainty of one of Alice's input bits. This comes at the cost of increasing the amount of qubits shared per OT~\cite{Lindell2007}. An example of these techniques applied in the context of the noisy storage model (where the commitment based check is replaced by a time delay under noisy memories) can be found in~\cite{PhysRevA.81.052336}. 

\subsection{Oblivious key distribution}
In order to make the quantum implementation of OT more practical during SMC we introduce the concept of oblivious keys. The protocol $\pi_{QOT}$ can be separated in two phases: the {\em Oblivious Key Distribution phase} which consists of steps 1 to 7 and forms the $\pi_{OKD}$ subprotocol, and the {\em Oblivious Transfer phase} which takes steps 8 to 10 and we denote as the $\pi_{OK \rightarrow OT}$ subprotocol. Note that after step 7 of $\pi_{QOT}$ the subsets $I_0,I_1$ have not been revealed to Alice, so she has no information yet on how the correlated and uncorrelated bits between $k$ and $\Tilde{k}$ are distributed (recall that $k$ and $\Tilde{k}$ are the result of removing the tested bits from the strings $s$ and $\Tilde{s}$ respectively). On the other hand, after receiving Alice's preparation bases, Bob does know the distribution of correlated and uncorrelated bits between $k$ and $\Tilde{k}$, which is recorded in the string $x$ ($x_i=0$ if $a_i = \Tilde{a}_i$, otherwise $x_i=1$). 
Note that until step 7 of the protocol all computation is independent of the input bits $e_0, e_1, c$. Furthermore, from step 8, only the strings $k, \Tilde{k}$, and $x$ are needed to finish the protocol (in addition to the input bits). We call these three strings collectively an {\em oblivious key}, depicted in Figure~\ref{fig:OK1}. Formally, let Alice and Bob be two agents. Oblivious Key Distribution (OKD) is a service that outputs to Alice the string $k = k_1k_2 \ldots k_\ell$ and to Bob the string $\Tilde{k} = \Tilde{k}_1\Tilde{k}_2 \ldots \Tilde{k}_\ell$ together with the bit string $x = x_1x_2 \ldots x_\ell$, such that $k_i = \Tilde{k}_i$ whenever $x_i = 0$ and $\Tilde{k}_i$ does not give any information about $k$ whenever $x_i = 1$. All of the strings are chosen at random for every invocation of the service. A pair $(k, (\Tilde{k},x))$ distributed as above is what we call an {\em oblivious key pair}. Alice, who knows $k$, is referred to as the sender, and Bob, who holds $\Tilde{k}$ and $x$, is the receiver. In other words, when two parties share an oblivious key, the sender holds a string $k$, while the receiver has only approximately half of the bits of $k$, but knows exactly which of those bits he has.  
\begin{figure}[t]
    \begin{center}
        \includegraphics[width=0.6\linewidth]{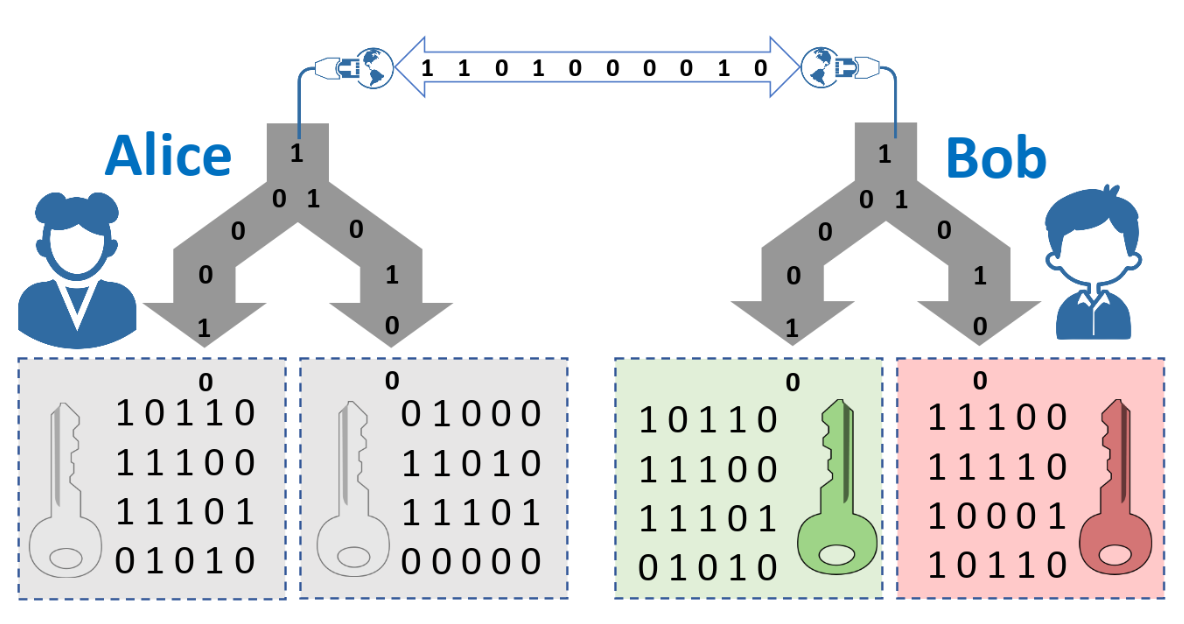}
        \caption{Oblivious keys. Alice has the string $k$ and Bob the string $\Tilde{k}$. For each party, the boxes in the left and right represent the bits of their string associated to the indices $i$ for which $x_i$ equals 0 (left box) or 1 (right box).
        Alice knows the entire key, Bob only knows half of the key, but Alice does not know which half Bob knows.}
        \label{fig:OK1}
    \end{center}
\end{figure}

When two parties have previously shared an oblivious key pair, they can securely produce OT by performing the steps $\pi_{OK \rightarrow OT}$ of $\pi_{QOT}$. This is significantly faster than current implementations of OT without any previous shared resource and does not require quantum communication during SMC. Note that the agents can perform, previously or concurrently, an OKD protocol to share a sufficiently large oblivious key, which can be then partitioned and used to perform as many instances of OT as needed for SMC.

Fortunately, it is possible to achieve fast oblivious key exchange if the parties have access to fast and reliable quantum communications and classical commitments. In order to use this QOT protocol, the commitment scheme must be instantiated. Consider the commitment protocol $\pi_{COMH}$ shown in Figure~\ref{fig:bcprotocol}, first introduced by Halevi and Micali. It uses a combination of universal and cryptographic hashing, the former to ensure statistical uniformity on the commitments, and the latter to hide the committed message. The motivation for the choice of this protocol for this task will become more apparent during the following sections as we discuss the security and efficiency characteristics of the composition of $\pi_{QOT}$ with $\pi_{COMH}$, henceforth referred as the $\pi_{HOK}$ (for Hybrid Oblivious Key) protocol for OT.

The existence of a reduction from OT to commitments, while proven within quantum cryptography through the $\pi_{QOT}$ protocol, is an open problem in classical cryptography. The existence of commitment schemes such as $\pi_{COMH}$, which do not rely on asymmetric cryptography, provides a way to obtain OT in the quantum setting while circumventing the disadvantages of asymmetric cryptography.

\begin{figure}
\noindent
\framebox{\parbox{\dimexpr\linewidth-2\fboxsep-2\fboxrule}{\small
  {\centering \textbf{Protocol} $ \pi_{COMH}$ \par}
  \textbf{Parameters:} Message length $\Tilde{n}$ and security parameter $k$. A universal hash family $\textbf{F}= \lbrace f: \lbrace 0,1 \rbrace^{\ell} \rightarrow \lbrace 0,1 \rbrace^{k} \rbrace$, with $\ell = 4k+2\Tilde{n}+4$. A collision resistant hash function $H$. \\
  \textbf{Parties:} The verifier Alice and the committer Bob. \\
  \textbf{Inputs:} Bob gets a string $\Tilde{m}$ of length $\Tilde{n}$. \\
  \textit{(Commit phase)}
  \begin{enumerate}[leftmargin=0.6cm]
      \item Bob samples $r \in \lbrace 0,1 \rbrace^{\ell}$, computes $y = H(r)$, and chooses $f \in \textbf{F}$, such that $f(r)=\Tilde{m}$. Then, he sends $(f,y)$ to Alice.
  \end{enumerate}
  \textit{(Open phase)}
  \begin{enumerate}[leftmargin=0.6cm]
      \setcounter{enumi}{1}
      \item Bob sends $r$ to Alice. \vspace{2mm}
      \item Alice checks that $H(r)=y$. If this test fails she aborts the protocol. Otherwise, she outputs $f(r)$.
  \end{enumerate}
  }}
  \caption{Commitment protocol based on collision resistant hash functions}
  \label{fig:bcprotocol}
\end{figure}

\section{Results and discussion}

\subsection{Security}
In this section, we analyse the security of the proposed composition of protocols. The main result is encapsulated in the following theorem.

\begin{theorem}
The protocol $\pi_{HOK}$ is secure as long as the hash function is collision resistant. Moreover, if the hash function models a Random Oracle, a simple modification of the protocol can make it universally-composable secure. 
\end{theorem} 
\begin{proof} The security proof relies on several well-established results in cryptography. First, notice that the $\pi_{HOK}$ protocol is closely related to the standard Quantum OT protocol $\pi_{QOT}$, which is proven statistically secure in Yao's original paper~\cite{Yao_1995} and later universally composable in the quantum composability framework~\cite{Unruh10}. The difference between the two is that $\pi_{QOT}$ uses ideal commitments, as opposed to the hash-based commitments in $\pi_{HOK}$. 
We start by showing that the protocol $\pi_{HOK}$ is standalone secure. For this case, we only need to replace the ideal commitment of $\pi_{QOT}$ with a standalone secure commitment protocol, such as the Halevi and Micali~\cite{Halevi1996}, which is depicted in $\pi_{COMH}$. Since the latter is secure whenever the hash function is collision resistant, we conclude that $\pi_{HOK}$ is secure whenever the hash function is collision resistant.

Finally, we provide the simple modification of $\pi_{HOK}$ that makes it universally-composable secure when the hash function models a Random Oracle. The modification is only required to improve upon the commitment protocol, as Yao's protocol with ideal commitments is universally-composable~\cite{Unruh10}. Indeed, we need to consider universally composable commitment scheme instead of $\pi_{COMH}$. This is achieved by the HMQ construction~\cite{Hofheinz04} which, given a standalone secure commitment scheme and a Random Oracle, outputs a universally-composable commitment scheme, which is perfectly hiding and computationally binding, that is, secure as far as collisions cannot be found. So we just need to replace $\pi_{COMH}$ with the output of the HMQ construction, when  $\pi_{COMH}$ and $H$ are given as inputs and $H$ models a Random Oracle.
\end{proof}

Regarding the above theorem we note that, for the composable security, the HMQ construction mentioned in the proof formally requires access to a random oracle, which is an abstract object used for studying security and cannot be realized in the real world. Hence, we leave it as an additional security property, as hash functions are traditionally modelled as random oracles. Stand alone security of the $\pi_{HOK}$ protocol \textit{does not} require the hash function to be a random oracle.

The use of collision resistant hash functions is acceptable in the quantum setting, as it has been shown that there exist functions for which a quantum computer does not have any significant advantage in finding collisions when compared with a classical one \cite{Aaronson:2004}. One point to note about the security of $\pi_{OKD}$ is that it is not susceptible to \textit{intercept now-decrypt later} style of attacks. Bob can attempt an attack in which he does not properly measure the qubits sent by Alice at step 2, and instead waits until Alice reveals the test subset in step 4 to measure honestly only those qubits. For that he must be able to control the openings of the commitment scheme such that Alice opens the values of his measurement outcomes for those qubits. In order to do this, he must be able find collisions for $H$ before step 5. This means that attacking the protocol by finding collisions of the hash function is only effective if it is done in real time, that is, between steps 3 and 5 of the protocol. This is in contrast to asymmetric cryptography based OT, in which Bob can obtain both bits if he is able to overcome the computational security at a later stage. 

Finally, we point out that the OT extension algorithms that are used during SMC often rely only on collision resistant hash functions~\cite{Asharov:2013:MEO:2508859.2516738} anyway. If those protocols are used to extend the base OTs produced by $\pi_{HOK}$, we can effectively speed up the OT rates without introducing any additional computational complexity assumption. 

\subsection{Efficiency}
Complexity-wise, the main problem with public-key based OT protocols is that they require a public/private key generation, encryption, and decryption per transfer. In the case of RSA and ElGamal based algorithms, this has complexity $O(n^{2.58})$ (where $N= 2^{n}$ is the size of the group), using Karatsuba multiplication and Berett reduction for Euclidian division~\cite{Menezes1996}. Post-quantum protocols are still ongoing optimization, but recent results show RLWE key genereration and encryption in time $O(n^2 \log(n))$~\cite{Ding2012}. 

To study the time complexity of the $\pi_{HOK}$ protocol, consider first the complexity of $\pi_{COMH}$. It requires two calls of $H$ and one call of the universal hash family $\textbf{F}$, $\Tilde{n}$ bit comparisons (if using the technique proposed in ~\cite{Halevi1996} to find the required $f$), and one additional evaluation of $f$. Cryptographic hash functions are designed so that their time complexity is linear on the size of the input, which in this case is $\ell=4k+2\Tilde{n}+4$. To compute the universal hashing, the construction in~\cite{Halevi1996} requires $\Tilde{n}k$ binary multiplications. Thus, the running time of $\pi_{COMH}$ is linear on the security parameter $k$. On the other hand, $\pi_{QOT}$ has two security parameters: $n$, associated to the size of the keys used to encrypt the transferred bits, and $m$, associated to the security of the measurement test done by Alice. The protocol requires $n+m$ qubit preparations and measurements, $n+m$ calls of the commitment scheme, and $n$ bit comparisons. This leads to an overall time complexity of $O(k(n+m))$ for the $\pi_{HOK}$ protocol, which is linear in all of its security parameters. 

In realistic scenarios, however, error correction and privacy amplification must be implemented during the $\pi_{OK \rightarrow OT}$. For the former, LDPC codes~\cite{Martinez2013} or the cascade algorithm~\cite{Brassard1993} can be used, and the latter can be done with universal hashing. For a given channel error parameter, these algorithms have time complexity linear in the size of the input string, which in our case is $n$. Hence, $\pi_{HOK}$ stays efficient when considering channel losses and preparation/measurement errors.

One of the major bottlenecks in the GMW protocol for SMC is the number of instances of OT required (it is worth noting that GMW uses 1-out-of-4 OT, which can efficiently be obtained from two instances of the 1-out-of-2 OT~\cite{naor05}). A single Advanced Encryption Standard (AES) circuit can be obtained with the order of $10^6$ instances of OT. However, with current solutions, i.e., with computational implementations of OT based on asymmetric classical cryptography, one can generate $\sim10^3$ secure OTs per second in standard devices ~\cite{cho:orl:15}. It is possible to use OT extension algorithms to increase its size up to rates of the order of $10^6$ OT per second~\cite{Asharov2017}. Several of such techniques are based on symmetric cryptography primitives~\cite{cho:orl:15}, such as hash functions, and could also be used to extend the OTs generated by $\pi_{HOK}$.

Due to the popularity of crypto-currencies, fast and efficient hashing machines have recently become more accessible. Dedicated hashing devices are able to compute SHA-256 at rates of $10^{12}$ hashes per second (see Bitfury, Ebit, and WhatsMiner, for example). In addition, existent standard Quantum Key Distribution (QKD) setups can be adapted to implement OKD, since both protocols share the same requirements for the generation and measurement of photons. Notably, QKD setups have already demonstrated secret key rates of the order of $10^{6}$ bits per second ~\cite{com:fro:luc:14,isl:lim:cah:17,ko:cho:cho:18, wan:pen:yin:18, pir:and:ber:19}. It is also worth mentioning that, as opposed to QKD, OKD is useful even in the case when Alice and Bob are at the same location. This is because in standard key distribution the parties trust each other and, if at the same location, they can just exchange hard drives with the shared key, whereas when sharing oblivious keys, the parties do not trust each other and need a protocol that enforces security. Thus, for the cases in which both parties being at the same location is not an inconvenience, the oblivious key rates can be further raised, as the effects of channel noise are minimized. 

Direct comparisons of OT generation speed between asymmetric cryptography techniques and quantum techniques are difficult because the algorithms run on different hardware. Nevertheless, as quantum technologies keep improving, the size and cost of devices capable of implementing quantum protocols will decrease and their use can result in significant improvements of OT efficiency, in the short-to-medium term future.

\section{Conclusions}
Motivated by the usefulness of SMC as a privacy-protecting data mining tool, and identifying its OT cost as its main implementation challenge, we have proposed a potential solution for practical implementation of OT as a subroutine SMC. The scheme consists on pre-sharing an oblivious key pair and then using it to compute fast OT during the execution of the SMC protocol. We call this approach hybrid because it uses resources traditionally associated with classical symmetric cryptography (cryptographic hash functions), as well as quantum state communication and measurements on conjugate observables, resources associated with quantum cryptography. The scheme is secure as far as the chosen hash function is secure against quantum attacks. In addition, we showed that the overall time complexity of $\pi_{HOK}$ is linear on all its security parameters, as opposed to the public-key based alternatives, whose time complexities are at least quadratic on their respective parameters. Finally, by comparing the state of current technology with the protocol requirements, we concluded that it has the potential to surpass current asymmetric cryptography based techniques.  

It was also noted that current experimental implementations of standard discrete-variable QKD can be adapted to perform $\pi_{HOK}$. The same post-processing techniques of error correction and privacy amplification apply, however, fast hashing subroutines should be added for commitments during the parameter estimation step. Future work includes designing an experimental setup, meeting the implementation challenges, and experimentally testing the speed, correctness, and security of the resulting oblivious key pairs. This includes computing oblivious key rate bounds for realistic scenarios and comparing them with current alternative technologies. Real world key rate comparisons can help us understand better the position of quantum technologies in the modern cryptographic landscape.

Regarding the use of quantum cryptography during the commitment phase; because of the impossibility theorem for unconditionally secure commitments in the quantum setting~\cite{Mayers97}, one must always work with an additional assumption on top of needing quantum resources. The noisy storage model provides an example in which the commitments are achieved by noisy quantum memories~\cite{lou:14,alm:etal:15,lou:16}. The drawback of this particular assumption is the fact that advances in quantum storage technology work against the performance of the protocol, which is not a desired feature. The added cost of using quantum communication is a disadvantage. So far, to the knowledge of the authors, there are no additional practical quantum bit commitment protocols that provide advantages in security or efficiency compared to classical ones once additional assumptions (such as random oracles, common reference strings, computational hardness, etc.,) are introduced. Nevertheless, we are optimistic that such protocols can be found in the future, perhaps by clever design, or by considering a different a kind of assumption outside of the standard ones.

\section*{Acknowledgments}
This work is supported by the Fundação para a Ciência e a Tecnologia (FCT) through national funds, by FEDER, COMPETE 2020, and by Regional Operational Program of Lisbon, under UIDB/50008/2020, UIDP/50008/2020, UID/CEC/00408/2013, POCI-01-0145-FEDER-031826, POCI-01-0247-FEDER-039728, PTDC/CCI-CIF/29877/2017,  PD/BD/114334/2016, PD/BD/113648/2015, and CEECIND/04594/2017.
\bibliographystyle{IEEEtran}
\bibliography{bibfile}

\end{document}